\documentclass[3p]{elsarticle}
\usepackage{amsmath,amssymb}
\usepackage{xcolor}
\usepackage{tikz}
\usetikzlibrary{matrix}
\usepackage{enumerate}
\usepackage{mathtools}

\journal{Nonlinear Analysis: Hybrid Systems}

\biboptions{authoryear}

\newcommand{\comment}[1]{}

\renewcommand{\qed}{\hfill \mbox{$\blacksquare$}}
\newcommand{\mer}{\hfill\ensuremath{\circ}}

\newcommand{\ki}{{\mathcal K}_\infty}

\def \R{{\mathbb R}}
\def \N{{\mathbb N}}

\def\K{\mathcal{K}}
\def\Ki{\mathcal{K}_{\infty}}
\def\C{\mathcal{C}}
\def\S{\mathcal{S}}

\def\KL{\mathcal{KL}}
\def\T{\mathcal{T}}

\def\U{\mathcal{U}}

\def\P{\mathcal{P}}
\def\comp{\,{\scriptstyle\circ}\,}
\def\id{\mathrm{id}}
\def\Ssup{\hat{\S}}
\def\Sup{\bar{\S}}
\def\Sdn{\underline{\S}}
\def\Sad{\S_{\textsc{ad}}}
\def\Srad{\S_{\textsc{rad}}}
\def\Sf{\S_{\textsc{F}}}
\def\Slim{\S_{\textsc{lim}}}
\def\usubset{\hspace{3pt}\mathrlap{\subset}\raisebox{1pt}{\hspace{2pt}$\scriptstyle u$}\hspace{4pt}}
\def\nn{\mathrm{n}}

\newtheorem{teo}{Theorem}[section]
\newtheorem{lem}[teo]{Lemma}
\newtheorem{defin}[teo]{Definition}
\newtheorem{prop}[teo]{Proposition}
\newtheorem{itremark}[teo]{Remark}

\newdefinition{example}[teo]{Example}
\newenvironment{remark}{\noindent\begin{itremark}\rm}{\end{itremark}}
\newproof{proof}{\textbf{Proof}}

\begin{document}

\begin{frontmatter}

\title{Uniform stability of nonlinear time-varying impulsive systems\\ with eventually uniformly bounded impulse frequency}


\author[ITBA]{Jos\'e L. Mancilla-Aguilar}
\ead{jmancill@itba.edu.ar}

\author[HHaddress]{Hernan Haimovich\corref{mycorrespondingauthor}}
\cortext[mycorrespondingauthor]{Corresponding author}
\ead{haimovich@cifasis-conicet.gov.ar}

\author[PFaddress]{Petro Feketa}
\ead{pf@tf.uni-kiel.de}

\address[ITBA]{Instituto Tecnol\'ogico de Buenos Aires, Av. E. Madero 399, Buenos Aires, Argentina.}
\address[HHaddress]{International Center for Information and Systems Science (CIFASIS), CONICET-UNR, Ocampo y Esmeralda, 2000 Rosario, Argentina.}
\address[PFaddress]{Kiel University, Chair of Automatic Control, Kaiserstra{\ss}e 2, 24143 Kiel, Germany}

\begin{abstract}
We provide novel sufficient conditions for stability of nonlinear and time-varying impulsive systems. These conditions generalize, extend, and strengthen many existing results. Different types of input-to-state stability (ISS), as well as zero-input global uniform asymptotic stability (0-GUAS), are covered by employing a two-measure framework and considering stability of both weak (decay depends only on elapsed time) and strong (decay depends on elapsed time and the number of impulses) flavors. By contrast to many existing results, the stability state bounds imposed are uniform with respect to initial time and also with respect to classes of impulse-time sequences where the impulse frequency is eventually uniformly bounded. We show that the considered classes of impulse-time sequences are substantially broader than other previously considered classes, such as those having fixed or (reverse) average dwell times, or impulse frequency achieving uniform convergence to a limit (superior or inferior). Moreover, our sufficient conditions are not more restrictive than existing ones when particularized to some of the cases covered in the literature, and hence in these cases our results allow to strengthen the existing conclusions.



\end{abstract}


\end{frontmatter}

\section{Introduction}
\label{sec:introduction}


The theory of impulsive systems \citep{SP87, LBS89} is a convenient mathematical framework for modeling processes that combine continuous and discontinuous behaviors. An impulsive system consists of an ordinary differential equation that governs the evolution of the state between jumps, a static law which introduces discontinuities at some isolated moments of time, and an impulse-time sequence which determines the instants when the static law comes into play. Applications of impulsive systems can be found in robotics \citep{tang2015tracking}, biomedical engineering \citep{rivadeneira2015observability}, population dynamics~\citep{rogovchenko1997nonlinear, yang2019recent}, and many other areas. The basis of the mathematical theory of impulsive systems as well as fundamental results on the existence and local stability of solutions are summarized in the monographs by \citet{SP87, LBS89, SaP95}.

Many recently developed methods for the global stability analysis of impulsive systems are tightly related to the notion of input-to-state stability~(ISS). ISS was introduced by \citet{Son89} for continuous-time systems with inputs and characterizes the behavior of solutions with respect to external disturbances. ISS of impulsive control systems was firstly studied in \cite{HLT05, HLT08} by providing a set of Lyapunov-based sufficient conditions that ensure ISS with respect to suitable classes of impulse-time sequences. These results were obtained by introducing the concept of exponential ISS-Lyapunov function. Two constants, called rate coefficients, are used to bound the evolution of the ISS-Lyapunov function along the trajectories of impulsive system during flows (constant~$c\in\mathbb R$) and jumps (constant~$d\in\mathbb R$). 
Relations called dwell-time conditions~(DTC) which restrict the number/frequency of jumps in order to guarantee ISS have been introduced. These conditions can be of two types: fixed or average dwell-time. The fixed dwell-time conditions utilize the minimum/maximum time-distance between two consecutive jumps. The average ones impose a bound for the number of jumps in some average sense. Sufficient conditions for ISS which are based on fixed dwell-time conditions are less widely applicable than the ones based on average dwell-time conditions. Some generalizations of the recently discussed approach, involving exponential ISS-Lyapunov functions with multiple and time-varying rate coefficients, have been proposed in~\citet{DF17, DaF16} and \citet{peng2018unified, peng2018lyapunov, ning2018indefinite}, respectively. In addition, the ideas of \citet{HLT05, HLT08} have been extended to some classes of time-delay \citep{DKMN12, r1, r3, wu2016input}, switched \citep{r1,liu2012class}, and stochastic \citep{ren2017stability, yao2014input, wu2016input} impulsive systems.

A more refined technique for the ISS and global stability analysis of impulsive control systems that relies on the concept of a candidate Lyapunov function with nonlinear rate functions
has been employed in \citet{liu2012class, DM13, ECC18, feketa2019average, Automatica2019, mancilla2019uniform}. The nonlinear rate functions 
are used to bound the evolution of the candidate Lyapunov function along the nonlinear flows and jumps of the impulsive system more precisely than what is possible by means of an exponential-type Lyapunov function 
\citep{HLT05, HLT08}. Hence, the corresponding sufficient conditions are supposed to be less conservative. A drawback of most sufficient conditions involving candidate Lyapunov functions with nonlinear rates was that these conditions are valid over impulse-time sequences having fixed dwell times. This may reduce applicability and thus degrade the benefits of employing nonlinear rates. To the best of authors' knowledge, the recent papers \citet{feketa2019average, Automatica2019} are the only works that provide sufficient conditions for the (nonuniform) ISS and global asymptotic stability (GAS) of nonlinear impulsive systems in terms of candidate Lyapunov functions with nonlinear rates over impulse-time sequences having average-type dwell-time bounds. 

The majority of results on stability of impulsive systems consider a bound on the state that decays as time elapses but is insensitive to the occurrence of jumps. In a time-varying setting, this stability notion is not robust and too weak to be meaningful, as shown in \citet{haiman_auto19c_arxiv}. By contrast, a stronger stability concept where the bound on the state also decays when jumps occur, as usually considered for hybrid systems \citep{caitee_scl09}, indeed is robust and more meaningful for impulsive systems in a time-varying setting \citep{haimovich2019strong}.

Motivated by \citet{feketa2019average, Automatica2019}, the contribution of this paper arises from the combination and improvement of the benefits of average-type dwell-time bounds and the most widely applicable uniform ISS results of \citet{mancilla2019uniform}. First, we significantly broaden the class of impulse-time sequences over which our stability results hold by considering impulse-time sequences having eventually uniformly bounded impulse frequency (see Section~\ref{sec:imp-time-class}). We show that many known classes of impulse-time sequences are uniform subsets (see Definition~\ref{def:usubset}) of the newly proposed classes (Lemma~\ref{lem:relats}). Then, we provide sufficient conditions for both weak and strong ISS of impulsive systems where the ISS bounds hold uniformly over initial time and over classes of impulse-time sequences. These conditions are based on Lyapunov-type functions having nonlinear rates. For increased generality, we formulate our results in a (time-varying) two-measure framework that incorporates different important stability notions unifyingly. Our results are thus stronger, less conservative and more widely applicable than many existing results. 
The remainder of the paper is organized as follows. In Section~\ref{sec:basic-definitions}, we introduce notation, the type of systems and the stability definitions considered. In Section~\ref{sec:imp-time-class}, we introduce different classes of impulse-time sequences and show the relationships between these and existing classes. The main results of the paper, consisting in most general Lyapunov-based sufficient conditions for stability, uniformly with respect to both initial time and impulse-time sequences within the considered classes, are provided in Section~\ref{sec:iss-under-eubif}. Application examples are provided in Section~\ref{sec:examples} and conclusions in Section~\ref{sec:conclusions}.

\section{Basic definitions}
\label{sec:basic-definitions}

\subsection{Notation}
\label{sec:notation}

The reals, positive, and nonnegative reals are denoted by $\R$, $\R_{>0}$ and $\R_{\ge 0}$, respectively. A function $\varphi : \R_{\ge 0} \to \R_{\ge 0}$ is said to be of class $\C_*$, written $\varphi\in\C_*$, if $\varphi$ is continuous and satisfies $\varphi(0)=0$. If $\varphi\in\C_*$, we write $\varphi\in\P$ if $\varphi(s) > 0$ for all $s>0$; $\varphi\in\K$ if $\varphi$ is strictly increasing; and $\varphi\in\Ki$ if $\varphi\in\K$ and $\varphi$ is unbounded. It is clear that $\Ki \subset \K \subset \P \subset \C_*$. If $\beta : \R_{\ge 0} \times \R_{\ge 0} \to \R_{\ge 0}$ is continuous, we write $\beta\in\KL$ if $\beta(\cdot,t) \in \K$ for all $t\ge 0$, $\beta(r,\cdot)$ is strictly decreasing whenever positive, and $\lim_{t\to\infty} \beta(r,t) = 0$ for all $r\ge 0$. For $a\in\R$, $\lfloor a \rfloor$ and $\lceil a \rceil$ denote, respectively, the greatest integer not greater than $a$ and the least integer not less than $a$. For any function $x: I\subset\R \to \R^n$, with $I$ an open interval, $x(t^-)$ and $x(t^+)$ denote respectively the left and right limits of $x$ at $t\in I$. For an infinite sequence $\{\Delta_k\}_{k=1}^\infty$ of real numbers, $\Delta_k \nearrow L$ means that the sequence is strictly increasing and that $\lim_{k\to\infty} \Delta_k = L$.

 \subsection{The system}
 Consider the time-varying impulsive system with inputs defined by the equations
 \begin{subequations}
   \label{eq:is}
   \begin{align}
     \label{eq:is-ct}
     \dot{x}(t) &=f(t,x(t),u(t)),\phantom{x(t^-)+g^-}\quad\text{for } t\notin \gamma,    \displaybreak[0] \\
     \label{eq:is-st}
     x(t) &=x(t^-)+g(t,x(t^-),u(t)),\phantom{f} \quad\text{for } t\in \gamma,
   \end{align}
 \end{subequations}
where $t\ge 0$, the state variable $x(t)\in \R^n$, the input variable $u(t)\in \R^m$, $f$ (the flow map) and $g$ (the jump map) are functions from $\R_{\ge 0}\times \R^n \times \R^m$ to $\R^n$, and $\gamma=\{\tau_k\}_{k=1}^{N} \subset (0,\infty)$, with $N$ finite or $N=\infty$ is the impulse-time sequence. By ``input'', we mean a Lebesgue measurable and locally essentially bounded function $u:[0,\infty)\to \R^m$; we denote by $\U$ the set of all the inputs. We only consider impulse-time sequences $\gamma=\{\tau_k\}_{k=1}^N$ that are strictly increasing and have no finite limit points, i.e. $\lim_{k\to \infty}\tau_k=\infty$ when the sequence is infinite; we employ $\Gamma$ to denote the set of all such impulse-time sequences. For any sequence $\gamma = \{ \tau_{k} \}^N_{k=1} \in \Gamma$ we define for convenience $\tau_0=0$ and $\tau_{N+1}=\infty$ when $N$ is finite; nevertheless, $\tau_0$ is never an impulse time, because $\gamma \subset (0,\infty)$ by definition.

We assume that for each input $u\in \U$ the map $f_u(t,\xi):=f(t,\xi,u(t))$ is a Carath\'eodory function. Hence the (local) existence of solutions of the differential equation $\dot x(t)=f(t,x(t),u(t))$ is ensured \citep[see][Thm.~I.5.1]{hale_book80}.

A solution of (\ref{eq:is}) corresponding to initial time $t_0\ge 0$, initial state $x_0\in \R^n$, input $u \in \U$ and impulse-time sequence $\gamma$ is a  
function $x:[t_0,T_x)\to \R^n$ such that: 
\begin{enumerate}[a)]
\item $x(t_0)=x_0$; 
\item $x$ is locally absolutely continuous on each nonempty interval of the form $J_k=[\tau_k,\tau_{k+1}) \cap [t_0,T_x)$ and $\dot{x}(t)=f(t,x(t),u(t))$ for almost all $t\in J_k$; and \label{item:solflow}
\item for all $\tau_k \in (t_0,T_x)$, the left limit $x(\tau_k^-)$ exists and is finite, and it happens that \label{item:soljump}
$$x(\tau_k) = x(\tau_k^-)+g(\tau_k,x(\tau_k^-),u(\tau_k)).$$
\end{enumerate}
Note that \ref{item:solflow}) implies that for all $t\in [t_0,T_x)$, $x(t)=x(t^+)$, i.e. $x$ is right-continuous at $t$.
 
The solution $x$ is said to be maximally defined if no other solution $y:[t_0,T_y)\to \R^n$ satisfies $y(t) = x(t)$ for all $t\in [t_0,T_x)$ and has $T_y > T_x$. We use $\T(t_0,x_0,u,\gamma)$ to denote the set of maximally defined solutions of (\ref{eq:is}) corresponding to initial time $t_0$, initial state $x_0$, input $u$ and impulse-time sequence $\gamma$. Since solutions locally exist but are not necessarily unique, the set $\T(t_0,x_0,u,\gamma)$ is nonempty but may contain more than one solution.

Note that even if $t_0 \in \gamma$, any solution $x\in\T(t_0,x_0,u,\gamma)$ begins its evolution by ``flowing'' and not by ``jumping''. This is because in item~\ref{item:soljump}) above, the time instants where jumps occur are those in $\gamma \cap (t_0,T_x)$.

\subsection{Input-to-state stability (ISS)}
\label{sec:input-state-stab}

From~(\ref{eq:is-st}), it is clear that the input values at impulse instants may instantaneously affect the state trajectory. Suitable stability properties must hence take the latter fact into consideration. For a given input $u \in \U$ and impulse-time sequence $\gamma \in \Gamma$, we thus consider the following bound over an interval $I\subset \R_{\ge 0}$:
\begin{align}
  \| u_I \|_{\gamma} &:= \max\left\{ {\text{ess sup}}_{t\in I} |u(t)| , \sup_{t\in \gamma\cap I} |u(t)| \right\}. \label{eq:iss-norm}
\end{align}
This definition is in agreement with that employed in \citet{caitee_cdc05,caitee_scl09} in the context of hybrid systems.

For greater generality, we formulate the stability properties in the framework of two measures \citep[see][]{liu2012class,chalib_jco06}. Let $\mathcal{H}$ be the set of functions $h:\R_{\ge 0}\times \R^n\to \R_{\ge 0}$. For $\gamma \in \Gamma$ and $t>s\ge 0$, let $n^\gamma_{(s,t]}$ be the number of impulse-time instants contained in $(s,t]$, that is
\begin{align}
  n^\gamma_{(s,t]} := \# \Big[ \gamma \cap (s,t] \Big].
\end{align}
\begin{defin}
  \label{def:stab}
  Let $h^o, h\in \mathcal{H}$ and $\S\subset \Gamma$. We say that the impulsive system (\ref{eq:is}) is
  \begin{itemize}
  \item weakly $(h^o,h)$-ISS over $\S$ if there exist $\beta \in \KL$ and $\rho \in \ki$ such that for all $t_0\ge 0$, $x_0\in \R^n$, $u\in \U$, $\gamma \in \S$ and $x\in \T(t_0,x_0,u,\gamma)$, it happens that for all $t\in [t_0,T_x)$,
    \begin{align}
      \label{eq:cwiss}
      h(t,x(t)) &\le \beta \left (h^o(t_0,x_0),t-t_0\right )
                  +\rho(\|u_{(t_0,t]}\|_{\gamma});
    \end{align}
  \item strongly $(h^o,h)$-ISS over $\S$ if there exist $\beta \in \KL$ and $\rho \in \ki$ such that for all $t_0\ge 0$, $x_0\in \R^n$, $u\in \U$, $\gamma \in \S$ and $x\in \T(t_0,x_0,u,\gamma)$, it happens that for all $t\in [t_0,T_x)$,
    \begin{align}
      \label{eq:ciss}
      h(t,x(t)) \le \beta \left (h^o(t_0,x_0),\ t-t_0+n^\gamma_{(s,t]}\right )
      + \rho(\|u_{(t_0,t]}\|_{\gamma}).
    \end{align}
  \end{itemize}
\end{defin}
By suitable selection of $h^0$ and $h$, one can recover the definitions of different stability properties usually considered in the analysis of both impulsive and nonimpulsive systems. For example, with $h^0(t,x)=h(t,x)=|x|$, the weak $(h^o,h)$-ISS property becomes the standard ISS property considered in the literature of systems with inputs. By considering, in addition, that the set where the inputs take values is $\R^0:=\{0\}$, then the standard definition of global uniform asymptotic stability (GUAS) for systems without inputs is recovered. By taking $h^0(t,x)=|x|$ we obtain an extension of the input-to-output stability property (IOS) studied in \citet{sonwan_siamjco00}; see \citet{liu2012class} for more examples. 

The decaying term in a weak property is insensitive to jumps, whereas that of a strong property forces additional decay whenever a jump occurs. The weak ISS property is the one considered in most of the literature on impulsive systems with inputs, whereas strong ISS is in agreement with the ISS property for hybrid systems as in \citet{libnes_tac14}. The weak stability properties are, however, not robust in the context of time-varying systems \citep[see][]{haiman_auto19c_arxiv}.

\section{Classes of impulse-time sequences}
\label{sec:imp-time-class}

Our main interest is to provide stability results that hold uniformly over both initial time and broad classes of impulse-time sequences. To this aim, we consider several classes of sequences involving upper bounds (Section~\ref{sec:class-up-bound}) and lower bounds (Section~\ref{sec:class-low-bound}) on the number of impulses. In Section~\ref{sec:unif-rel}, we provide preliminary results establishing properties and relationships between the classes.

\subsection{Classes involving upper bounds}
\label{sec:class-up-bound}

\begin{defin}\label{def:UIB}
  A family of impulse-time sequences $\S \subset \Gamma$ is said to be uniformly incrementally bounded \citep[UIB, ][]{haiman_aadeca18} if there exists a continuous and nondecreasing function $\phi:\R_{>0} \to \R_{\ge 0}$ such that $n^\gamma_{(s,t]} \le \phi(t-s)$ for all $t>s\ge 0$, and all $\gamma\in\S$.
\end{defin}
Every sequence contained in some UIB family $\S$ is such that the number of impulses occuring in a period of fixed finite duration cannot become infinite as the initial time for such a period becomes increasingly large. A specific subfamily of the UIB class is that of average dwell-time sequences with chatter bound $\nn\in\N$ and average dwell time $\tau>0$, defined as
\begin{equation*}
\Sad(\nn, \tau) = \left\{\gamma\in\Gamma: n^\gamma_{(s,t]} \leq \nn + \frac{t-s}{\tau},\forall \;0\leq s\leq t\right\}.
\end{equation*}
It is clear that $\Sad(\nn,\tau)$ is UIB for every $\nn\in\N$ and $\tau>0$. Other classes are more easily described by referring to the impulse frequency (i.f.) $\frac{n^\gamma_{(s,s+t]}}{t}$ instead of the dwell time. Given $\rho\ge 0$ we consider the class of sequences having i.f. eventually uniformly upper bounded by $\rho$:
\begin{align}
  \Sup(\rho) &= \Bigg\{ \gamma\in\Gamma: \forall \varepsilon > 0, \exists T = T(\gamma,\varepsilon) > 0 \text{ s.t. } 
  \label{eq:Supu-def}
      \frac{n^\gamma_{(s,s+t]}}{t} \le \rho + \varepsilon, \forall t \ge T, \forall s\geq 0 \Bigg\}. 
\end{align}
We will later show (Lemma~\ref{lem:relats}) that the class $\Sup(\rho)$ is sufficiently broad for our purposes. Given a sequence $\gamma \in \Sup(\rho)$ it straightforwardly follows that $\limsup_{t\to \infty}\frac{n^{\gamma}_{(s,s+t]}}{t}=L_s\le \rho$ for all $s\ge 0$. Moreover, $L_s=L_0$ for all $s$. In fact, for $s\ge 0$ and $t>0$ we have that
 \begin{align*}
 \frac{n^\gamma_{(s,s+t]}}{t}=\frac{n^\gamma_{(0,s+t]}}{s+t}. \frac{s+t}{t}-\frac{n^\gamma_{(0,s]}}{t},
 \end{align*}
and then taking $\limsup$ as $t$ goes to $\infty$ it follows that 
\begin{align*}
 L_s=\limsup_{t\to \infty}\frac{n^\gamma_{(s,s+t]}}{t}=\limsup_{t\to \infty}\frac{n^\gamma_{(0,s+t]}}{s+t}=L_0.
\end{align*}
However, the condition $\limsup_{t\to \infty}\frac{n^{\gamma}_{(0,t]}}{t}\le \rho$ is not sufficient for ensuring that $\gamma \in \Sup(\rho)$, as Example~\ref{ex:ev-unif-ub} shows.
\begin{example}
  \label{ex:ev-unif-ub}
  Consider $\gamma \in \Gamma$ formed by the concatenation of an infinite number of finite sequences $\{\tau_k^1\}_{k=0}^2$, $\{\tau_k^2\}_{k=0}^2$, $\{\tau_k^3\}_{k=0}^5$, \ldots, $\{\tau_k^\ell\}_{k=0}^{p_\ell-1}$, \ldots. The finite sequences $\{\tau_k^\ell\}$ are defined as follows:
  \begin{align*}
    \tau_k^1 &:= 1 + \frac{k}{2} & k &= 0,1,2;\\
    \tau_k^\ell &:= 2^\ell - 1 + \frac{k}{p_\ell - 1} & k &= 0,\ldots,p_\ell -1,\quad \ell = 2, 3, \ldots;\\
    p_\ell &= 3\cdot 2^{\ell-2}.
  \end{align*}
  For each $\ell \ge 2$, the finite sequence $\{\tau_k^\ell\}_{k=0}^{p_\ell-1}$ is strictly increasing and contains $p_\ell = 3\cdotp 2^{\ell-2}$ equally spaced elements within an interval of length $1$; this follows by evaluating $\tau_0^\ell = 2^\ell - 1$, $\tau_{p_\ell - 1}^\ell = 2^\ell$ and $\tau_{k+1}^\ell - \tau_k^\ell = \frac{1}{p_\ell - 1} =: \Delta_\ell$. Note also that $3 = \tau_0^2 > \tau_2^1 = 2$ and that for $\ell \ge 2$, we have $\tau_0^{\ell+1} = 2^{\ell+1} - 1 > \tau_{p_\ell - 1}^\ell = 2^\ell$. Therefore, the concatenation $\gamma = \{ \{\tau_k^1\},\{\tau_k^2\},\ldots \}$ yields a strictly increasing sequence with no finite limit points and hence $\gamma \in \Gamma$. Let $t\ge 2$ and $\ell=\lceil \log_2(t) \rceil$, then
  \begin{align*}
    \frac{n^\gamma_{(0,t]}}{t}\le \frac{n^\gamma_{(0,2^\ell]}}{2^{\ell}-1} &= \frac{3\cdotp 2^{\ell-1}}{2^\ell-1} 
  \end{align*}
and therefore $\limsup_{t\to \infty}\frac{n^\gamma_{(0,t]}}{t}\le \frac{3}{2}$. 
Since given any finite length $T>0$ and arbitrarily large number $N$, we can always find $s\ge 0$ such that $n^\gamma_{(s,s+T]} \ge N$, it follows that $\gamma \notin \Sup(\rho)$ for any $\rho>0$.\mer
\end{example}

A sufficient condition for $\gamma\in \Sup(\rho)$ is that 
$$\limsup_{t\to \infty}\frac{n^\gamma_{(s,s+t]}}{t}=L$$
uniformly w.r.t. (u.w.r.t.) $s\ge 0$ with $L\le \rho$ (see Lemma \ref{lem:relats}). We thus define
\begin{align*}
  \Ssup(L) &:= \left\{ \gamma\in\Gamma: \limsup_{t\to \infty}\frac{n^{\gamma}_{(s,s+t]}}{t}=L,\;\text{u.w.r.t.}\;s\ge 0\right \}.
\end{align*}

\subsection{Classes involving lower bounds}
\label{sec:class-low-bound}

We can have some analogous definitions involving lower bounds, such as reverse average dwell time with reverse chatter bound $\nn\in\N$ and reverse average dwell time $\tau > 0$:
\begin{align*}
  \Srad&(\nn, \tau) = \left\{\gamma\in\Gamma: n^\gamma_{(s,t]} \geq -\nn + \frac{t-s}{\tau}, \forall\: 0\leq s\leq t\right\},
  \end{align*}
and the class of sequences having i.f. eventually uniformly lower bounded by $\rho$:
\begin{align}
  \Sdn&(\rho) = \Bigg\{ \gamma\in\Gamma: \forall \varepsilon > 0, \exists T = T(\gamma,\varepsilon)> 0 \text{ s.t. } 
  \label{eq:Sdnu-def}
      \frac{n^\gamma_{(s,s+t]}}{t} \ge \rho - \varepsilon, \forall t \ge T, \forall s\geq 0 \Bigg\}. 
\end{align}
In this case, a sufficient condition for $\gamma\in \Sdn(\rho)$ is that 
$$\liminf_{t\to \infty}\frac{n^\gamma_{(s,s+t]}}{t}=L$$
u.w.r.t. $s\ge 0$ with $L\ge \rho$ (see Lemma \ref{lem:relats}) and we define
\begin{align*}
  \check{\S}(L) &:= \left\{ \gamma\in\Gamma: \liminf_{t\to \infty}\frac{n^{\gamma}_{(s,s+t]}}{t}=L,\;\text{u.w.r.t.}\;s\ge 0\right \}.
\end{align*}

\subsection{Classes involving upper and lower bounds}
\label{sec:up-low-bnd-class}

For comparison of our results, we define the class of impulse-time sequences with fixed dwell times: 
\begin{align*}
  \Sf(\theta_1,\theta_2) = \Big\{ \gamma = \{\tau_k\}_{k=1}^\infty
  \in\Gamma: \theta_1\leq \tau_{k+1}-\tau_k\leq\theta_2, \ \forall
  k\in\N\Big\}.
\end{align*}
and that of sequences having i.f. with a uniform limit $\rho$:
\begin{align}
  \Slim^u(\rho) = \Bigg\{ \gamma\in\Gamma: \lim_{t\to \infty}\frac{n^\gamma_{(s,s+t]}}{t}=\rho\;\text{u.w.r.t.}\;s\ge 0\Bigg\}.
\end{align}
It is clear that $\Slim^u(\rho) \subset \Sup(\rho) \cap \Sdn(\rho)$ and hence is a much smaller class than either $\Sup(\rho)$ or $\Sdn(\rho)$.

\subsection{Uniformity and relationships}
\label{sec:unif-rel}

The word ``uniformly'' in the definitions of
$\Sup(\rho)$ and $\Sdn(\rho)$ refers to uniformity with respect to initial time. In the definition of UIB, however, ``uniformly'' refers to uniformity with respect to both initial time and every impulse-time sequence in the set. We are also interested in sets of impulse-time sequences where the bounds imposed by $\Sup(\rho)$ or $\Sdn(\rho)$ hold uniformly over every sequence in the set. We therefore employ the following definition.
\begin{defin}
  \label{def:usubset}
  We say that a set of impulse-time sequences $\S \subset \Gamma$ is a uniform subset of $\Sup(\rho)$ (resp. $\Sdn(\rho)$), and write $\S\usubset \Sup(\rho)$ (resp. $\S\usubset \Sdn(\rho)$), if $\S \subset \Sup(\rho)$ (resp. $\S \subset \Sdn(\rho)$) and $T$ in (\ref{eq:Supu-def}) (resp. (\ref{eq:Sdnu-def})) can be selected independently of $\gamma\in\S$.
\end{defin}
\begin{remark}
  Every $\S \subset \Sup(\rho)$ ($\S \subset \Sdn(\rho)$) containing a finite number of sequences satisfies $\S \usubset \Sup(\rho)$ ($\S \usubset \Sdn(\rho)$) because the required $T$ can be taken as
  $\max_{\gamma\in\S} T(\gamma,\varepsilon)$. \mer
\end{remark}

In Lemma~\ref{lem:relats}, we prove that the classes $\Sup(\rho)$ and 
$\Sdn(\rho)$ are very broad, i.e., many known classes of impulse-time sequences are (uniform) subsets of the classes with eventually uniformly upper/lower bounded i.f. We require the following definition.
\begin{defin}
  \label{def:delta-pert}
  Given $\Delta\ge 0$, we say that $\gamma^*=\{\tau_k^*\}_{k=1}^N\in \Gamma$ ($N\in\N$ or $N = \infty$) is a $\Delta$-perturbation of $\gamma=\{\tau_k\}_{k=1}^N\in \Gamma$ if $\tau_k^*\in [\tau_k-\Delta,\tau_k+\Delta]$ for all $k$. 
\end{defin}
\begin{lem}
  \label{lem:relats}
  The following relationships hold.
  \begin{enumerate}[i)]
  \item $\Sad(\nn,\tau)\usubset \Sup\left(\frac{1}{\tau}\right)$ and $\Srad(\nn,\tau)\usubset \Sdn\left(\frac{1}{\tau}\right)$.
    \label{item:rel-euup-sad}
  \item If $\S\usubset \Sup(\rho)$ ($\S\usubset \Sdn(\rho)$), then for each $0<\tau<\frac{1}{\rho}$ ($\tau>\frac{1}{\rho}$) there exists $\nn$ such that
    $\S\subset \Sad(\nn,\tau)$ ($\S\subset \Srad(\nn,\tau)$).\label{item:rel-uib-sad}
  \item $\Ssup(L) \subset \Sup(\rho)$ ($\check{\S}(L)\subset\Sdn(\rho)$) for all $L\le \rho$ ($L\ge \rho$). \label{item:limADT-euup}
  \item $\Sf(\theta_1,\theta_2)\usubset \Sup\left(\frac{1}{\theta_1}\right)$ and $\Sf(\theta_1,\theta_2)\usubset \Sdn\left(\frac{1}{\theta_2}\right)$.
    \label{item:fixed-euup}
  \item $\Sup(\rho)$ and $\Sdn(\rho)$ are persistent under perturbations:\label{item:pers-pert}
    \begin{enumerate}[a)]
      \makeatletter
      \renewcommand{\p@enumii}{\theenumi)}
      \makeatother
    \item if $\gamma^*$ is a $\Delta$-perturbation of $\gamma\in \Sup(\rho)$ ($\gamma\in\Sdn(\rho)$), then $\gamma^*\in \Sup(\rho)$ ($\gamma^*\in \Sdn(\rho)$);\label{item:Dperta}
    \item if $\S\usubset \Sup(\rho)$ ($\S\usubset\Sdn(\rho)$), then $\S_{\Delta}$, the set of all the $\Delta$-perturbations of sequences in $\S$, satisfies $\S_{\Delta} \usubset \Sup(\rho)$ ($\S_{\Delta} \usubset \Sdn(\rho)$).\label{item:Dpert}
    \end{enumerate}
  \end{enumerate}
\end{lem}
\begin{proof}
\ref{item:rel-euup-sad}) Consider any $\gamma\in\Sad(\nn,\tau)$. Then, $$n^\gamma_{(s,t]} \leq \nn + \frac{t-s}{\tau}\quad \text{for all }0\leq s < t.$$
Therefore,
$$\frac{n^\gamma_{(s,s+t]}}{t} \leq \frac{\nn}{t} + \frac{1}{\tau}\quad \forall s \geq 0, \forall t> 0.$$
Given $\varepsilon>0$, let $T = T(\varepsilon) := \frac{\nn}{\varepsilon}$. Then, $\frac{n^\gamma_{(s,s+t]}}{t} \leq \frac{1}{\tau} + \varepsilon$ holds for all $s \geq 0$ and $t\geq T$. Hence $\gamma\in \Sup\left(\frac{1}{\tau}\right)$. Since $T$ does not depend on the specific $\gamma$ considered, we conclude that $\Sad(\nn,\tau)\usubset \Sup\left(\frac{1}{\tau}\right)$. The inclusion $\Srad(\nn,\tau)\usubset \Sdn\left(\frac{1}{\tau}\right)$ can be proved analogously.

\ref{item:rel-uib-sad}) 
By assumption, for every $\varepsilon>0$ there exists $T=T(\varepsilon)>0$ such that $\frac{n^\gamma_{(s,s+t]}}{t} \le \rho + \varepsilon$ holds for all $t \ge T$, $s\geq 0$ and $\gamma\in \S$. For $0\le t<T$, it holds that
$$
n^\gamma_{(s,s+t]}\leq n^\gamma_{(s,s+T]}\leq (\rho+\varepsilon)T \quad \forall s\geq 0, \forall\gamma\in\S.
$$
Combining the bounds for small and large $t$ yields
\begin{align}
  \label{eq:ntotbnd}
  n^\gamma_{(s,s+t]}\leq (\rho+\varepsilon)T + (\rho+\varepsilon)t\quad \forall s,t\ge 0, \forall \gamma\in\S.
\end{align}
Let $0<\tau<\frac{1}{\rho}$, set $\varepsilon=1/\tau-\rho > 0$ and take the corresponding $T=T(\varepsilon)$. Substituting these quantities into (\ref{eq:ntotbnd}), it follows that $n^\gamma_{(s,s+t]} \le T/\tau + t/\tau$. Taking $\nn = \lceil T/ \tau \rceil$, then $\S \subset \Sad(\nn,\tau)$. The other case can be proved analogously.

\ref{item:limADT-euup}) Let $\gamma \in \Ssup(L)$ with $L\le \rho$ and $\varepsilon > 0$. By definition of $\Ssup(L)$, 
there exists $T=T(\gamma,\varepsilon)>0$ such that 
$$ L-\varepsilon \le \sup_{\tau \ge t} \frac{n^{\gamma}_{(s,s+\tau]}}{\tau}\le L+\varepsilon\quad \forall s\ge 0, \forall t\ge T.$$
Then $\frac{n^{\gamma}_{(s,s+t]}}{t}\le \rho+\varepsilon$ for all $t\ge T$ and hence $\gamma \in \Sup(\rho)$. The other inclusion can be proved analogously.

\ref{item:fixed-euup}) This follows from item \ref{item:rel-euup-sad}) and the facts that $\Sf(\theta_1,\theta_2)\subset \Sad(1,\theta_1)$ and  $\Sf(\theta_1,\theta_2)\subset \Srad(1,\theta_2)$. 

\ref{item:Dperta}) Pick any $\gamma\in\Sup(\rho)$ and let $\gamma^*$ be a $\Delta$-perturbation of $\gamma$. Let $\varepsilon_1 > 0$ and define $\varepsilon = \varepsilon_1/2$. Since $\gamma\in\Sup(\rho)$ then there exists $T=T(\gamma, \varepsilon)>0$ such that
\begin{equation}\label{eq:lemmaVi}
\frac{n^\gamma_{(s,s+t]}}{t} \le \rho + \varepsilon \quad \forall t \ge T, \forall s\geq 0.
\end{equation}
Let us pick some $t^*>\max\{\Delta, T\}$. Then,
$$\frac{n^\gamma_{(s,s+t^*]}}{t^*} = \frac{n^\gamma_{(s,s+\Delta]}+n^\gamma_{(s+\Delta,s+t^*]}}{t^*}\le \rho + \varepsilon \ \forall s\geq 0$$
and, hence,
$$n^\gamma_{(s,s+\Delta]} \le (\rho + \varepsilon)t^*-n^\gamma_{(s+\Delta,s+t^*]} \le (\rho + \varepsilon)t^* \ \forall s\geq 0.$$
For any $s,t\geq 0$ denote by $n^+_{(s,s+t]}=n^{\gamma^*}_{(s,s+t]}-n^{\gamma}_{(s,s+t]}$. By definition of $\Delta$-perturbation, it follows that
$$ |n^+_{(s,s+t]}| \le 2\sup_{r\ge 0} n^\gamma_{(r,r+\Delta]} \quad \forall s,t\ge 0. $$
 Then, adding $\frac{n^+_{(s,s+t]}}{t}$ to the left and right-hand sides of \eqref{eq:lemmaVi}, we get
\begin{align*}
\frac{n^{\gamma^*}_{(s,s+t]}}{t} &\le \rho + \varepsilon + \frac{n^+_{(s,s+t]}}{t} \le \rho + \varepsilon + \frac{2\sup\limits_{r\geq 0}n^\gamma_{(r,r+\Delta]}}{t} \le \rho + \varepsilon + \frac{2(\rho+\varepsilon)t^*}{t}\quad \forall t \ge T, \forall s\geq 0.
\end{align*}
Let $T_1=\frac{2(\rho+\varepsilon)t^*}{\varepsilon_1-\varepsilon} = \frac{4(\rho+\varepsilon)t^*}{\varepsilon_1} > 0$. Then, $\frac{n^{\gamma^*}_{(s,s+t]}}{t} \le \rho + \varepsilon_1$ for all $t \ge T_1$ and $s\geq 0$; hence, $\gamma^* \in \Sup(\rho)$. 

\ref{item:Dpert}) This follows from the proof of \ref{item:Dperta}) by noting that if $\gamma\in \S\usubset \Sup(\rho)$, then $T$ and hence $t^*$ can be chosen independently of $\gamma$. This causes $T_1$ to become independent of $\gamma^* \in \S_\Delta(\rho)$. As a consequence, $\S_{\Delta} \usubset \Sup(\rho)$. 

The implications involving $\Sdn(\rho)$ can be proved analogously.
\qed
\end{proof}
By means of Lemma~\ref{lem:relats}, we may see that $\Sup(\rho)$ and $\Sdn(\rho)$ incorporate impulse-time sequences with
\begin{itemize}
\item (reverse) average dwell-time \citep{HLT05,HLT08}, Lemma~\ref{lem:relats}\ref{item:rel-euup-sad}-\ref{item:rel-uib-sad});
\item i.f.~eventually uniformly (upper/lower) convergent
\citep{DF17, feketa2019average}, Lemma~\ref{lem:relats}\ref{item:limADT-euup}); 
\item fixed dwell-time \citep{SP87, DM13}, Lemma~\ref{lem:relats}\ref{item:fixed-euup});
\item non-fixed impulse-time moments within predefined time-windows \citep{Tan2015, feng2017linear, Automatica2019}, Lemma~\ref{lem:relats}\ref{item:pers-pert}).
\end{itemize}

In order to gain more insight into the breadth of the classes $\Sup(\rho)$ and $\Sdn(\rho)$, and the facts given by Lemma~\ref{lem:relats}\ref{item:rel-euup-sad}-\ref{item:rel-uib-sad}), we remark that for every $\tau>0$, the sets $\Sup \left (\frac{1}{\tau}\right )$ and $\Sdn \left (\frac{1}{\tau}\right )$ contain sequences which do not belong to, respectively, $\Sad(\nn,\tau)$ and $\Srad(\nn,\tau)$ for any $\nn \in \N$. Example~\ref{ex:Sup-gen} illustrates this fact for the case $\tau = 1$. Examples for arbitrary values of $\tau$ can be obtained by simple modifications.
\begin{example}
  \label{ex:Sup-gen}
Consider $\gamma=\{\tau_k\}_{k=1}^{\infty}$ defined by
\begin{align*}
  \tau_1 = 1\quad\text{and}\quad \tau_k = k - \sum_{\ell=2}^k \frac{1}{\ell}\quad \text{for }k\ge 2,
\end{align*}
so that $\Delta_k := \tau_k-\tau_{k-1}=1-1/k$ for all $k\ge 2$. Let $\varepsilon>0$ and define $\tau':=\frac{1}{1+\varepsilon/2}$. Note that $0<\tau'<1$. Since $\Delta_k\nearrow 1$, then there exists $k_0\ge 2$ so that $1>\Delta_k\ge \tau'$ for all $k\ge k_0$. Set $n_0 :=n^{\gamma}_{(0,k_0]}$ and $T = T(\varepsilon) := 2\frac{n_0+1}{\varepsilon}$. For every $s\ge k_0$, we have that $\tau_k\ge s$ implies that $k\ge k_0$.
In consequence, for all $s\ge k_0$ and $t> 0$
\begin{align*}
 n^{\gamma}_{(s,s+t]}\le 1+\frac{t}{\tau'},
\quad\text{and then}\quad
 \frac{n^{\gamma}_{(s,s+t]}}{t}\le \frac{1}{t}+\frac{1}{\tau'}.
\end{align*}
For $0\le s<k_0$ and $t> 0$,
\begin{align*}
 \frac{n^{\gamma}_{(s,s+t]}}{t} \le \frac{n^{\gamma}_{(0,k_0+t]}}{t} \le \frac{n_0+n^{\gamma}_{(k_0,k_0+t]}}{t}\le \frac{n_0+1}{t}+\frac{1}{\tau'}.
\end{align*}
Therefore, for all $s\ge 0$ and $t > 0$ it follows that 
\begin{align*}
 \frac{n^{\gamma}_{(s,s+t]}}{t}\le \frac{n_0+1}{t}+\frac{1}{\tau'}.
\end{align*}
Recalling the definitions of $\tau'$ and $T$, we arrive to
\begin{align*}
 \frac{n^{\gamma}_{(s,s+t]}}{t}\le 1+\varepsilon\quad \forall s\ge 0,t\ge T.
\end{align*}
We have thus shown that $\gamma \in \Sup(1)$. Next, we prove that for every $\nn\in \N$, $\gamma \notin \Sad(\nn,1)$. Suppose for a contradiction that  $\gamma \in \Sad(\nn^*,1)$ for some $\nn^*\in\N$. By definition of $\Sad$ then
\begin{align}
  \label{eq:sad1}
 n^{\gamma}_{(0,t]}\le \nn^*+t\quad \text{for all }t>0.
\end{align}
Let $k\in \N$ with $k \ge 2$. Since $\tau_k=k-\sum_{\ell=2}^k\frac{1}{\ell}$ then 
\begin{align}
  \label{eq:sad2}
 n^{\gamma}_{(0,k]}= k+n^{\gamma}_{(\tau_k,k]}\ge k+\left \lfloor\sum_{\ell=2}^k\frac{1}{\ell} \right \rfloor.
\end{align}
Consider $k$ sufficiently large so that $\sum_{\ell=2}^k\frac{1}{\ell} > \nn^* + 1$ and let $t = k$. Then, (\ref{eq:sad1}) gives $n^\gamma_{(0,k]} \le k + \nn^*$ but (\ref{eq:sad2}) gives $n^\gamma_{(0,k]} \ge k + \nn^* + 1$. This is clearly a contradiction. \mer  
\end{example}
Arguments similar to those used in Example~\ref{ex:Sup-gen} show that the sequence $\gamma=\{\tau_k\}_{k=1}^{\infty}$ defined recursively by $\tau_1=1$ and $\tau_{k}=k + \sum_{\ell=2}^k (1/\ell)$ for $k\ge 2$ satisfies $\gamma \in \Sdn(1)$ and $\gamma \notin \Srad(\nn,1)$ for any $\nn\in \N$. 

The following example illustrates a class $\S \subset \Gamma$ that is a subset of $\Sup(1)$ but not a uniform subset.
\begin{example}
  \label{ex:S-not-unif}
  For each $\nn\in\N$ with $\nn \ge 2$, consider a sequence $\gamma_\nn \in \Gamma$ constructed by concatenating an infinite number of finite sequences $\{\tau_{\nn,k}^\ell\}_{k=0}^{\nn-1}$, for $\ell = 1, 2, \ldots,$ whose elements are defined as follows
  \begin{align*}
    \tau_{\nn,k}^\ell = 1 + (\ell-1)\nn + \frac{k}{\ell-1+\nn}.
  \end{align*}
  For $0\le k \le \nn-2$, we have
  \begin{align*}
    \tau_{\nn,k+1}^\ell - \tau_{\nn,k}^\ell = \frac{1}{\ell-1+\nn}
  \end{align*}
  so that the spacing between consecutive elements in the finite sequence $\{\tau_{\nn,k}^\ell\}_{k=0}^{\nn-1}$ becomes smaller and smaller as $\ell \to \infty$ because $\lim_{\ell\to\infty} \tau_{\nn,k+1}^\ell - \tau_{\nn,k}^\ell =0$. It can be shown that $\gamma_\nn \in \Sad(\nn,1)$. From Lemma~\ref{lem:relats}\ref{item:rel-euup-sad}), then $\gamma_\nn \in \Sup(1)$ for every $\nn$. Therefore, the class $\S := \{\gamma_\nn : \nn\in\N, \nn \ge 2\}$ satisfies $\S \subset \Sup(1)$. However, $\S$ is not a uniform subset of $\Sup(1)$ because, from the proof of Lemma~\ref{lem:relats}\ref{item:rel-euup-sad}) then the required $T$ in (\ref{eq:Supu-def}) cannot be taken independent of $\nn$ and hence cannot be independent of $\gamma_\nn$. \mer
\end{example}

\section{ISS under eventually uniformly bounded i.f.}
\label{sec:iss-under-eubif}

In this section, we provide sufficient conditions for the weak and strong $(h^o,h)$-ISS over classes of impulse-time sequences having eventually uniformly bounded impulse frequency. Our main results are stated in Section~\ref{sec:main-res}. In Section~\ref{sec:proof-tech}, we explain our proof technique, based on the analysis of a comparison system. The remaining technical results required for the full proof are given in Section~\ref{sec:rem-proofs}.

\subsection{Main results}
\label{sec:main-res}
We say that a locally Lipschitz function $V:\R_{\ge 0}\times \R^n\to \R$ is a $(h^o,h)$-ISS Lyapunov function candidate for system (\ref{eq:is}) if 
\begin{enumerate}[a)]
  \item there exist $\phi_1,\phi_2\in \Ki$ so that for all $t\ge 0$ and $\xi \in \R^n$,\label{item:Vphi12}
    \begin{align} 
      \label{eq:bound1}
      \phi_1(h(t,\xi))\le V(t,\xi)\le \phi_2(h^o(t,\xi));
    \end{align}
  \item \label{item:Vimply} there exist $\chi, \pi \in \Ki$, a locally integrable function $p:\R_{\ge 0}\to \R_{\ge 0}$, a continuous function $\varphi:\R_{\ge 0}\to \R$ and $\psi \in \P$ such that for all $t\ge 0$, $\xi\in \R^n$ and $\mu \in \R^m$, 
  \begin{enumerate}[i)] 
   \item \label{item:lyapbound} $D^{+}_{f}V(t,\xi,\mu) \le -p(t)\varphi(V(t,\xi))$ if $V(t,\xi)\ge \chi(|\mu|)$;
   \item \label{item:lyapbound2}
      $V(t,\xi+g(t,\xi,\mu)) \le \psi(V(t,\xi) )$ if $V(t,\xi)\ge \chi(|\mu|)$;
      \item  $V(t,\xi+g(t,\xi,\mu)) \le \pi(|\mu|)$ if $V(t,\xi)\le \chi(|\mu|)$.
  \end{enumerate}
\end{enumerate}
Here, $D^{+}_{f}V(t,\xi,\mu)$ denotes the upper-right Dini derivative of $V$ along $f$ at $(t,\xi,\mu)\in \R_{\ge 0}\times \R^n\times \R^m$, i.e.
\begin{align*}
 D^{+}_{f}V(t,\xi,\mu):=\limsup_{h \to 0^+}\frac{V(t+h,\xi+h f(t,\xi,\mu))-V(t,\xi)}{h}.
 \end{align*}
\begin{teo}
  \label{thm:main} Let $V$ be a $(h^o,h)$-ISS Lyapunov function candidate for system (\ref{eq:is}) with $\psi,\varphi \in \P$. Suppose that there exists $\theta>0$ such that
  \begin{align} 
    \label{eq:int}
    M:=\sup_{a>0}\int_a^{\psi(a)}\frac{ds}{\varphi(s)}< \inf_{t\ge 0}\int_{t}^{t+\theta}p(s)ds=:N
  \end{align}
  Then, the following hold.
  \begin{enumerate}[a)]
  \item If $M>0$, then (\ref{eq:is}) is strongly $(h^o,h)$-ISS over any $\S \subset \Gamma$ such that $\S\usubset \Sup\left(\frac{1}{\theta} \right)$.
  \item If $M= 0$, then (\ref{eq:is}) is weakly $(h^o,h)$-ISS over $\S=\Gamma$ and strongly $(h^o,h)$-ISS over any UIB family $\S\subset \Gamma$.
  \item If $M< 0<N$, then (\ref{eq:is}) is strongly $(h^o,h)$-ISS over $\S=\Gamma$.
  \end{enumerate}
\end{teo}
Theorem~\ref{thm:main} gives sufficient conditions for $(h^o,h)$-ISS when the continuous part of the dynamics cannot be destabilizing for large values of the state. The case $M<0$ corresponds to stabilizing impulses and hence the ensuing stability is very strong. The most interesting case is $M>0$, corresponding to destabilizing impulses. Inversely to Theorem~\ref{thm:main}, Theorem~\ref{thm:main2} addresses the case when the continuous part cannot be stabilizing.
\begin{teo} \label{thm:main2}
  Let $V$ be a $(h^o,h)$-ISS Lyapunov function candidate for system (\ref{eq:is}) with $\psi,-\varphi\in\P$. Let $\S\subset\Gamma$ satisfy $\S\usubset \Sdn\left(\frac{1}{\theta} \right)$ with $\theta > 0$ such that 
  \begin{align}  \label{eq:integralr2}
    \inf_{a>0} \int_{\psi(a)}^a\frac{ds}{- \varphi(s)} > \sup_{t\ge 0}\int_{t}^{t+\theta}p(s)\:ds.
  \end{align}
  Suppose that, in addition, 
  \begin{align} \label{eq:nofet}
    \int_1^\infty \frac{ds}{-\varphi(s)} = \infty. 
  \end{align}
  Then (\ref{eq:is}) is strongly $(h^o,h)$-ISS over $\S$.
\end{teo}


Theorems~\ref{thm:main} and \ref{thm:main2} provide a manifold extension of the results for the stability analysis of impulsive systems available in the literature. More precisely, these theorems extend the results in \citet{mancilla2019uniform,Automatica2019,feketa2019average} to provide two-measure uniform ISS results for time-varying nonlinear impulsive systems over sequences having eventually uniformly bounded impulse frequency. 

The results in \citet[Sections~4, 5]{mancilla2019uniform}, in turn, provide sufficient conditions for two-measure ISS that constituted a substantial extension of previously available results. The reader may refer to \citet{mancilla2019uniform} for specific explanations on how existing results were extended by the latter. In particular, Theorem~4.4 in \citet{mancilla2019uniform} provides sufficient conditions for the $(h^o,h)$-ISS of a system of the form~(\ref{eq:is}) based on a Lyapunov function candidate of the type considered here but where ISS holds uniformly only over sequences with fixed dwell times. The extension from fixed dwell-time sequences to sequences having eventually uniformly bounded i.f.~is substantial, as shown by Lemma~\ref{lem:relats}. In addition, excepting for the class of impulse-time sequences, the sufficient conditions of Theorems~\ref{thm:main} and~\ref{thm:main2} coincide with those of Theorem~4.4 in \citet{mancilla2019uniform}, and hence in this case the current results strengthen the corresponding conclusions. 

The results in \citet{Automatica2019} provide nonuniform GAS results over sequences having nonfixed impulse times over predefined time windows and those in \citet{feketa2019average} provide nonuniform ISS results over sequences having a uniform limit, i.e. sequences $\gamma \in \Slim^u(L)$ for some $L>0$. The extension with respect to the latter results regards uniformity with respect to initial time and the (broader) class of impulse-time sequences considered, two-measure ISS, and the consideration of Lyapunov function candidates with possibly time-varying rates.


\subsection{Proof technique}
\label{sec:proof-tech}

\hyphenation{diff-e-rence}
If $V$ is a Lyapunov function candidate for system~(\ref{eq:is}), we consider the following one-dimensional differential~/~difference inclusion system, which we henceforth call comparison system:
\begin{subequations}
  \label{eq:isc}
  \begin{align}
    \label{eq:is-ctc}
    \dot{z}(t) &\in(-\infty,-p(t)\varphi(z(t))], & 
    t\notin \gamma,  \\  
    \label{eq:is-stc}
    z(t) &\in[0,\psi(z(t^-))], & 
    t\in \gamma.
  \end{align}
\end{subequations}
We say that a function $z:I_z \to \R_{\ge 0}$, with $I_z=[t_0,T_z)$ is a solution of (\ref{eq:isc}) corresponding to $\gamma =\{\tau_k\}\in \Gamma$, initial time $t_0 \ge 0$ and initial condition $z_0\ge 0$ if i) $z(t_0) = z_0$, ii) for every nonempty interval $J_k = [\tau_k,\tau_{k+1}) \cap I_z$, $z$ is locally absolutely continuous on $J_k$ and $\dot{z}(t)\le -p(t)\varphi(z(t))$ for almost all $t\in J_k$, and iii) for every $\tau_k\in \gamma \cap (t_0,T_z)$, it happens that $z(\tau_k^-)$ exists and $0\le z(\tau_k) \le \psi(z(\tau_k^-))$. A solution $z$ of (\ref{eq:isc}) is maximally defined if it does not have a proper extension; it is forward complete if $T_z=\infty$. We will employ $\C(t_0,z_0,\gamma)$ to denote the set of maximally defined solutions $z$ of \eqref{eq:isc} corresponding to $\gamma \in \Gamma$, initial time $t_0$ and initial condition $z_0$. Note that, by definition, every solution of the comparison system is nonnegative.

We say that the comparison system (\ref{eq:isc}) is weakly or strongly GUAS (uniformly) over $\S \subset \Gamma$ if there exists a function $\beta \in \KL$ such that every $z\in\C(t_0,z_0,\gamma)$ with $\gamma \in \S$, $t_0\ge 0$ and $z_0\ge 0$ satisfies, respectively
\begin{align} 
  \label{eq:wbeta}
  \text{(weak)}\quad z(t) &\le \beta\left (z_0,t-t_0\right) &\forall t\in I_z,\\
  \label{eq:sbeta}
  \text{(strong)}\quad z(t) &\le \beta\left (z_0,t-t_0+n^\gamma_{(t_0,t]} \right) &\forall t\in I_z.
\end{align}
Section~3 of \citet{mancilla2019uniform} gives the basis method for ensuring that system (\ref{eq:is}) is weakly or strongly $(h^o,h)$-ISS over $\S$ by analyzing the comparison system (\ref{eq:isc}) given by means of an $(h^o,h)$-ISS Lyapunov function candidate. Specifically, Theorem~3.1 in \citet{mancilla2019uniform} states that system~(\ref{eq:is}) is weakly or strongly $(h^o,h)$-ISS over some arbitrary class $\S$ provided the corresponding comparison system is, respectively, weakly or strongly GUAS uniformly over $\S$. In the current paper, our main proof technique consists in establishing the required type of GUAS for the comparison system and applying Theorem~3.1 in \citet{mancilla2019uniform} to ensure the required type of $(h^o,h)$-ISS. Propositions~\ref{prop:main} and~\ref{prop:main2} below provide the required GUAS results and hence constitute our main technical contribution. Their proofs are given in Section~\ref{sec:rem-proofs}.
\begin{prop}
  \label{prop:main} Let $\psi,\varphi \in \P$ and $p:\R_{\ge 0}\to\R_{\ge 0}$ be locally integrable. Let $\theta>0$ be such that (\ref{eq:int}) holds. Then, the following hold.
  \begin{enumerate}[a)]
   \item If $M>0$, then the comparison system (\ref{eq:isc}) is strongly GUAS uniformly over any $\S\usubset \Sup\left(\frac{1}{\theta}\right)$.\label{item:Mgt0}
   \item If $M= 0$, then the comparison system (\ref{eq:isc}) is weakly GUAS uniformly over $\S=\Gamma$ and strongly GUAS uniformly over any UIB family $\S\subset \Gamma$.\label{item:Meq0}
    \item If $M< 0 < N$, then the comparison system (\ref{eq:isc}) is strongly GUAS uniformly over $\S=\Gamma$.\label{item:Mlt0}
  \end{enumerate}
\end{prop}
\begin{prop} \label{prop:main2}
  Let $\psi,-\varphi \in \P$, $p:\R_{\ge 0}\to \R_{\ge 0}$ be locally integrable, and $\theta>0$ be such that (\ref{eq:integralr2})--(\ref{eq:nofet}) are satisfied. Then, (\ref{eq:isc}) is strongly GUAS uniformly over $\S$ for every $\S\subset\Gamma$ such that $\S\usubset \Sdn^{u}\left (\frac{1}{\theta}\right )$.
\end{prop}

\begin{proof}[Theorems \ref{thm:main} and \ref{thm:main2}]
 As previously mentioned, the proofs of Theorems~\ref{thm:main} and~\ref{thm:main2} follow straightforwardly from Propositions~\ref{prop:main} and~\ref{prop:main2}, respectively, and Theorem~3.1 in \citet{mancilla2019uniform}. Employing the notation and definitions in \citet{mancilla2019uniform}, consider the parametrized family of impulsive systems $\{\Sigma_{\gamma}\}_{\gamma \in \S}$, where $\Sigma_{\gamma} = (\gamma,f_\gamma,g_\gamma)$, with $f_\gamma=f$ and $g_\gamma=g$ for all $\gamma\in \S$. Here, we employ the family $\S$ of impulse-time sequences as the parameter set and the sequence $\gamma$ as a parameter \citep[see Section~II.B in][]{mancilla2019uniform}. Also consider the family of functions $\{V_{\gamma}\}_{\gamma\in \S}$ with $V_{\gamma}=V$ for all $\gamma \in \S$, where $V$ is the $(h^o,h)$-ISS Lyapunov function candidate appearing in the hypotheses of both Theorem~\ref{thm:main} and~\ref{thm:main2}. Then, $\{V_{\gamma}\}_{\gamma\in \S}$ satisfies Assumption~1 in \citet{mancilla2019uniform}, and the family of comparison systems associated with $\{\Sigma_{\gamma}\}_{\gamma \in \S}$ and $\{V_{\gamma}\}_{\gamma\in \S}$ \citep[defined in eq.~(7) of ][]{mancilla2019uniform} coincides with the comparison system (\ref{eq:isc}). Applying Theorem~3.1 of \citet{mancilla2019uniform} to the family of systems $\{\Sigma_{\gamma}\}_{\gamma\in \S}$, then the system (\ref{eq:is}) is weakly or strongly $(h^o,h)$-ISS over $\S$ when the comparison system (\ref{eq:isc}) is weakly or strongly GUAS over $\S$, respectively. Finally, the weak or strong GUAS over $\S$ of the comparison system is established in Propositions~\ref{prop:main} and~\ref{prop:main2} in each of the considered cases. \qed 
\end{proof}

\subsection{Remaining proofs}
\label{sec:rem-proofs}

It can be shown, following the lines of Proposition~2.5 in \cite{linson_jco96}, that (\ref{eq:isc}) is weakly GUAS uniformly over $\S$ if and only if the following hold
\begin{enumerate}[ci)]
 \item (Global uniform stability, GUS) there exists $\alpha\in \K$ such that \label{item:gusdef}
 \begin{align}
\label{eq:GUS}
|z(t)|\le \alpha(|z(t_0)|)\quad \forall t\in I_z,
\end{align}
for all $t_0\ge 0$, $z_0\ge 0$, $\gamma \in \S$ and $z\in \C(t_0,z_0,\gamma)$;
\item (Uniform attractivity) for all $0< \varepsilon \le R$ there exists $T=T(\varepsilon,R)>0$ such that for all $t_0\ge 0$, $z_0\ge 0$ with $|z_0|\le R$, $\gamma \in \mathcal{S}$ and $z\in \C(t_0,z_0,\gamma)$ we have that $|z(t)|\le \varepsilon$ for all $t \in [t_0+T,\infty)\cap I_z$. \label{item:uadef}
\end{enumerate}

The following result can be proved in the same manner as Proposition 2.3 in \cite{mancilla2019uniform}.
\begin{lem} \label{lem:equiv}
 Let $\S\subset \Gamma$ be UIB. Then (\ref{eq:isc}) is strongly GUAS over $\S$ if and only if it is weakly GUAS over $\S$.  
\end{lem}

The proof of Proposition \ref{prop:main} requires the following results.
\begin{lem} 
  \label{lem:eqguas}
  Let $\varphi\in\P$ and $p:\R_{\ge 0}\to\R_{\ge 0}$ be locally integrable. Suppose there exists $\theta>0$ such that 
  \begin{align} \label{eq:infimum}
   N:=\inf_{t\ge 0}\int_t^{t+\theta}p(s)ds>0.
  \end{align}
Then, the system 
\begin{align} \label{eq:cseq}
 \dot w=-p(t)\varphi(w),
\end{align}
has the following properties.
\begin{enumerate}[a)]
 \item For every $t_0\ge 0$ and $w_0\ge 0$, there exists a unique forward-in-time solution $w_{t_0,w_0}:[t_0,\infty)\to \R_{\ge 0}$ of (\ref{eq:cseq}) such that $w_{t_0,w_0}(t_0)=w_0$.
 \item There exists $\beta\in \KL$ such that
 \begin{align}
  w_{t_0,w_0}(t)\le \beta(w_0,t-t_0)\quad \forall t\ge t_0,\;w_0\ge 0.
 \end{align}
\end{enumerate}
\end{lem}
\begin{proof}
a) Consider the continuous and increasing function $F:(0,\infty)\to (a,b)$
\begin{align}
 F(r)=\int_{1}^{r}\frac{ds}{\varphi(s)}
\end{align}
with $a=\lim_{r\to 0^+}F(r)$ and $b=\lim_{r\to \infty}F(r)$. $F$ is 
bijective, and therefore its inverse $F^{-1}:(a,b)\to (0,\infty)$ is continuous and increasing. 

Let $w$ be a solution of (\ref{eq:cseq}) such that $w(t_0)=w_0$. Since $w$ is nonincreasing and bounded from below, there is no finite escape time $t\ge t_0$ and therefore $w$ is defined on $[t_0,\infty)$. If $w(t)>0$ on an interval $[t_0,t_1)$, then a simple calculation shows that $w(t)=F^{-1}(F(w_0)-\int_{t_0}^tp(s)ds)$ for all $t\in [t_0,t_1)$. If for some $t_1\ge t_0$, $w(t_1)=0$, then the facts that $0\le w(t)$ for all $t\ge t_1$ and that $w$ is nonincreasing imply that $w(t)=0$ for all $t\ge t_1$. This shows the uniqueness forward-in-time of the solutions of (\ref{eq:cseq}).  

b) It suffices to show that the nonnegative solutions of (\ref{eq:cseq}) satisfies c\ref{item:gusdef}) and c\ref{item:uadef}). Let $t_0\ge 0$, $w_0\ge 0$ and $w_{t_0,w_0}:[t_0,\infty)\to \R_{\ge 0}$ be the unique solution of (\ref{eq:cseq}) such that $w(t_0)=w_0$. Since $w_{t_0,w_0}$ is nonincreasing, we have that $w_{t_0,w_0}(t)\le w_0$. In consequence c\ref{item:gusdef}) holds with $\alpha(r)\equiv r$. 

As for c\ref{item:uadef}), let $0<\varepsilon\le R$. Let $s\ge 0$ be so that $F^{-1}(F(R)-s)=\varepsilon$, pick $m_0\in \N$ so that $m_0N>s$, where $N$ is the constant appearing in (\ref{eq:infimum}), and define $T=m_0\theta$. Let $t_0\ge 0$ and $w_0\ge 0$. We claim that there is a $t_0\le t_1\le t_0+T$ such that $w_{t_0,w_0}(t_1)\le \varepsilon$. If $w_{t_0,w_0}(t)> \varepsilon$ for all $t\in [t_0, t_0+T]$, then $w_{t_0,w_0}(t)=F^{-1}(F(w_0)-\int_{t_0}^{t_0+T}p(s)ds)\le F^{-1}(F(R)-s)=\varepsilon$, since $\int_{t_0}^{t_0+T}p(s)ds\ge mN\ge s$, which is absurd. Then, since $w_{t_0,w_0}$ is nonincreasing, it follows that $w_{t_0,w_0}(t)\le \varepsilon$ for all $t\ge t_1$. \qed
\end{proof}
\begin{lem} 
  \label{lem:ufc}
  Let $\varphi, \psi \in \P$ and $p:\R_{\ge 0} \to \R_{\ge 0}$ be locally integrable. Then, maximally defined solutions of the comparison system (\ref{eq:isc}) are forward complete. Moreover, if $z\in\C(t_0,z_0,\gamma)$ with $t_0 \ge 0$, $z_0 \ge 0$ and $\gamma\in\Gamma$ satisfies $z(s) = 0$ for some $s\ge t_0$, then $z(t) = 0$ for all $t\ge s$. 
\end{lem}
\begin{proof}
Let $z\in \C(t_0,z_0,\gamma)$ with $t_0\ge 0$, $z_0\ge 0$ and $\gamma \in \Gamma$. Suppose that $[t_0,T_z)$, with $T_z<\infty$ is the maximal interval of definition of $z$. Since $z$ is nonincreasing between consecutive impulse times and bounded from below, there exists $z(T_z^-)$. If $T_z\notin \gamma$, then there exists $\delta>0$ so that $[T_z,T_z+\delta)\cap \gamma=\emptyset$. Define $z^*(t)=z(t)$ if $t\in [t_0,T_z)$ and $z^*(t)=w^*(t)$ for $t\in [T_z,T_z+\delta)$, where $w^*$ is any solution of (\ref{eq:cseq}) such that $w^*(T_z)=z(T_z^-)$. If $T_z\in \Gamma$, take $\delta>0$ so that $(T_z,T_z+\delta)\cap \gamma=\emptyset$, and define $z^*$ as before, but with $w^*(T_z)=\psi(z(T_z^-))$. In both cases $z^*$ is a proper extension of $z$ which is solution of (\ref{eq:isc}), which is absurd. This show that $T_z=+\infty$. Suppose that $z(s)=0$ for some $s\ge t_0$. If $s_0=s$ and $s_1<s_2<\cdots$ are the impulse times in $\gamma\cap(s,\infty)$, we have that $z(t)=0$ for all $t\in [s_0,s_1)$ since $z$ is nonincreasing and nonnegative. Taking into account that $0\le z(s_1)\le \psi(z(s_1^-))=\psi(0)=0$, it follows that $z(t)=0$ for all $t\in [s_1,s_2)$. By proceeding in this way, it follows that $z(t)=0$ for all $t\ge s_0$. \qed 
\end{proof}
\begin{proof}[Proposition \ref{prop:main}]
  Let $\delta=N-M>0$. 

  \ref{item:Mgt0}) Let $M>0$ and $\S\usubset\Sup\left(\frac{1}{\theta}\right)$. 
  By~(\ref{eq:Supu-def}) and Definition~\ref{def:usubset}, for each $\varepsilon>0$ there is a $T(\varepsilon)>0$ such that 
\begin{equation}
  \label{eq:NSupineq}
  \frac{n^\gamma_{(s,s+t]}}{t} \leq \frac{1}{\theta} + \varepsilon \quad \forall t\geq T(\varepsilon), \forall s\geq 0, \forall\gamma \in \S.
\end{equation}
Let $\varepsilon_0= \frac{\delta}{2 \theta M}$. Pick $m_0\in \N$ such that $m_0\theta=:T_0\ge T(\varepsilon_0)$ and let $k_0=\lfloor (\frac{1}{\theta}+\varepsilon_0)T_0\rfloor$. Then, $n^\gamma_{(s,s+T_0]} \le k_0$ for all $s\ge 0$ and all $\gamma \in \S$. Pick any $\bar\psi \in \Ki$ satisfying $\max\{\psi,\id\} \le \bar\psi$, and define $\eta\in\Ki$ as $\eta=\bar{\psi}^{k_0}$, i.e.
\begin{align*}
  \eta := \underbrace{\bar\psi \comp \cdots \comp \bar\psi}_{k_0\text{ times}}.
\end{align*}
Note that if $k\in \N_0$ and $k\le k_0$, then\footnote{We define $\bar\psi^0 = \id$.} $\bar{\psi}^{k}\le \eta$.

Let $t_0 \ge 0$, $z_0 \ge 0$, $\gamma\in \S$ and $z\in \C(t_0, z_0, \gamma)$. Due to Lemma \ref{lem:ufc}, $z$ is defined for all $t\ge t_0$ and if $z(s)=0$ for some $s\ge t_0$ then $z(t)=0$ for all $t\ge s$. Taking into account that $z$ is nonincreasing between consecutive impulse times, that $\psi(z(t^{-}))\le \bar{\psi}(z(t^{-}))$ at each $t\in \gamma$, and that for all $t\in [t_0,t_0+T_0]$, $k=n^\gamma_{(t_0,t]} \le k_0$, it follows that $z(t)\le \bar{\psi}^k(z_0)\le \eta(z_0)$ for all $t\in [t_0,t_0+T_0]$; hence,
\begin{align}
  \label{eq:zTbnd2}
  \sup_{t_0 \le t \le t_0+T_0} z(t) \le \eta(z_0).
\end{align}

If $\gamma \cap (t_0,t_0+T_0]=\emptyset$, then $\dot z(t) \le -p(t)\varphi(z(t))$ from~\eqref{eq:is-ctc} implies that
\begin{align}\label{eq:b1}
 \int_{z(t_0+T_0)}^{z(t_0)}\frac{ds}{\varphi(s)}\ge \int_{t_0}^{t_0+T_0}p(t)dt\ge m_0N.
\end{align}
If $\gamma \cap (t_0,t_0+T_0]\neq \emptyset$, consider the increasing sequence $\{t_i\}_{i=1}^{k} := \gamma \cap (t_0,t_0+T_0]$. Note that $k\le k_0$. Suppose that $z(t_0+T_0)>0$. Then $z(t)>0$ for all $t\in [t_0,t_0+T_0]$ due to Lemma \ref{lem:ufc}. From~\eqref{eq:isc}, (\ref{eq:int}) and the definitions of $\varepsilon_0$, $m_0$ and $T_0$ it follows that
\begin{align}
  \int\limits_{z(t_0+T_0)}^{z(t_0)} \frac{ds}{\varphi(s)} &= \sum_{i=0}^{k-1} \int\limits_{z(t_{i+1})}^{z(t_i)} \frac{ds}{\varphi(s)} + \int\limits_{z(t_0+T_0)}^{z(t_k)}\frac{ds}{\varphi(s)} 
                                   = \sum_{i=0}^{k-1} \left( \int\limits_{z(t_{i+1}^-)}^{z(t_i)} \frac{ds}{\varphi(s)} + \int\limits_{z(t_{i+1})}^{z(t_{i+1}^-)} \frac{ds}{\varphi(s)} \right)+\int\limits_{z(t_0+T_0)}^{z(t_k)}\frac{ds}{\varphi(s)}\notag\\
  &\ge \sum_{i=0}^{k-1} \int\limits_{z(t_{i+1}^-)}^{z(t_i)} \frac{ds}{\varphi(s)} + \int\limits_{z(t_0+T_0)}^{z(t_k)}\frac{ds}{\varphi(s)}+ \sum_{i=1}^{k} \int\limits_{\psi(z(t_{i}^-))}^{z(t_{i}^-)} \frac{ds}{\varphi(s)}  
  \ge \int_{t_0}^{t_0+T_0}p(t)dt - kM \label{eq:b2exp} \\
  &\ge m_0N-kM \ge m_0N-k_0M \ge m_0N-\left (\frac{1}{\theta}+\varepsilon_0 \right )T_0=\frac{m_0\delta}{2}>0, \label{eq:b2}
\end{align}
where the first term in the first inequality in (\ref{eq:b2exp}) follows from consideration of (\ref{eq:is-ctc}). Taking into account (\ref{eq:b1}) and (\ref{eq:b2}) it follows that if $z(t_0+T_0)> 0$ then 
\begin{align} \label{eq:b3}
 \int\limits_{z(t_0+T_0)}^{z(t_0)} \frac{ds}{\varphi(s)}\ge \frac{m_0\delta}{2}
 \end{align}
and, in consequence, $z(t_0+T_0)<z(t_0)$. Repeating the preceding reasoning on each of the intervals $[t_0+(\ell-1) T_0, t_0+\ell T_0]$, with $\ell \in \N$, we have that $0<z(t_0+\ell T_0)<z(t_0+(\ell-1) T_0)$ or $z(t_0+\ell T_0)=0$ and that $\sup_{t_0+(\ell-1)T_0 \le t \le t_0+\ell T_0} z(t) \le \eta(z(t_0+(\ell-1)T_0))$. Thus $z(t)\le \eta(z_0)$ for all $t\ge t_0$ and the comparison system (\ref{eq:isc}) is GUS according to c\ref{item:gusdef}).

Next, let $0<\varepsilon\le R$. Set $\bar\varepsilon := \eta^{-1}(\varepsilon)\le \varepsilon$ and pick $\ell\in \N$ such that 
\begin{align}
 \int_{\bar\varepsilon}^R \frac{ds}{\varphi(s)}\le \frac{1}{2}m_0\ell\delta. 
\end{align}
Let $z\in \C(t_0,z_0,\gamma)$ with $t_0\ge 0$, $z_0\ge 0$ and $\gamma\in \S$. Suppose that $z(t_0+\ell T_0)>0$. Then, from (\ref{eq:b3}) we have
\begin{align}
 \int_{z(t_0+\ell T_0)}^R\frac{ds}{\varphi(s)} &\ge \int_{z(t_0+\ell T_0)}^{z(t_0)}\frac{ds}{\varphi(s)}
 =\sum_{r=0}^{\ell-1} \int\limits_{z(t_0+(r+1) T_0)}^{z(t_0+r T_0)}\frac{ds}{\varphi(s)}\ge \frac{1}{2}m_0\ell \delta \ge \int_{\bar\varepsilon}^R \frac{ds}{\varphi(s)}, 
\end{align}
which implies that $z(t_0+\ell T_0)\le \bar\varepsilon$ and then that $z(t)\le \eta(\bar\varepsilon)=\varepsilon$ for all $t \ge t_0+\ell T_0$. If $z(t_0+\ell T_0)=0$ then $z(t)=0$ for all $t\ge t_0+\ell T_0$. Therefore, uniform attractivity according to c\ref{item:uadef}) follows with $T(\varepsilon,R)=\ell T_0$. We have thus established that the comparison system is weakly GUAS uniformly over $\S$. From Lemma~\ref{lem:relats}~\ref{item:rel-uib-sad}), then $\S$ is UIB. Since the comparison system (\ref{eq:isc}) is weakly GUAS uniformly over $\S$ and $\S$ is UIB, Lemma \ref{lem:equiv} that (\ref{eq:isc}) is strongly GUAS uniformly over $\S$. 

\ref{item:Meq0}) Let $M=0$. From~(\ref{eq:int}), then $\psi(r)\le r$ for all $r\ge 0$. Then, the result of 
Lemma~\ref{lem:eqguas} implies that the comparison system is weakly GUAS uniformly over $\Gamma$. To see this, let $z\in \C(t_0,z_0,\gamma)$, with $t_0\ge 0$, $z_0\ge 0$ and $\gamma\in \Gamma$. Applying well-known comparison results for ordinary differential equations, using in addition the fact that $\psi(r)\le r$ and that, due to the uniqueness of the solutions of (\ref{eq:cseq}), $w_{t_0,z_1}\le w_{t_0,z_0}$ if $0\le z_1\le z_0$, it follows that $z(t)\le w_{t_0,z_0}(t)$ for all $t\ge t_0$. In consequence, $z(t)\le \beta(z_0,t-t_0)$ for all $t\ge t_0$, where $\beta\in\KL$ is given by Lemma~\ref{lem:eqguas}. Applying Lemma \ref{lem:equiv}, then the comparison system is strongly GUAS uniformly over $\S$ for any UIB family $\S\subset \Gamma$ .

\ref{item:Mlt0}) 
Let $M<0<N$. From~(\ref{eq:int}), then $0<\psi(r)<r$ for all $r>0$. Define $\bar\psi(r)=\max_{0\le s\le r}\psi(s)$. Then $\bar \psi$ is continuous and nondecreasing, $0< \bar \psi(r)<r$ for all $r>0$ and $\bar \psi(0)=0$. Consider the following differential/difference inclusion system
\begin{subequations}
  \label{eq:isc4}
  \begin{align}
    \label{eq:is-ctc4}
    \dot{z}(t) &\in(-\infty,-p(t)\varphi(z(t))], & 
    t\notin \gamma,  \\  
    \label{eq:is-stc4}
    z(t) &\in[0,\bar \psi(z(t^-))], & 
    t\in \gamma.
  \end{align}
\end{subequations}
We have that (\ref{eq:isc}) is strongly GUAS uniformly over $\S$ if (\ref{eq:isc4}) is, since every solution of (\ref{eq:isc}) is also a solution of (\ref{eq:isc2}). Due to Remark~3 in \cite{mancilla2019uniform}, for checking that (\ref{eq:isc4}) is strongly GUAS over $\Gamma$, it suffices to only consider the solutions of the impulsive system
\begin{subequations}
  \label{eq:isc5}
  \begin{align}
    \label{eq:is-ctc5}
    \dot{z}(t) &=-p(t)\varphi(z(t)), & 
    t\notin \gamma,  \\  
    \label{eq:is-stc5}
    z(t) &=\bar \psi(z(t^-)), & 
    t\in \gamma,
  \end{align}
\end{subequations}
corresponding to initial times $t_0\ge 0$, initial conditions $z_0\ge 0$ and $\gamma \in \Gamma$. Consider the difference equation
\begin{align}\label{eq:diff}
w_{k+1}=\bar \psi(w_k),\quad w_0 \ge 0. 
\end{align} 
Since $\bar\psi(r)<r$ for all $r>0$, the discrete-time system is GUAS and there must exist $\beta_d \in \KL$ such that every solution of this difference equation satisfies $w_k\le \beta_d(w_0,k)$ for all $k\in \N_0$. 

Let $z$ be a solution of (\ref{eq:isc5}) with $\gamma \in \Gamma$ such that $z(t_0)=z_0$, with $t_0\ge 0$ $z_0\ge 0$. Let $t> t_0$ and $k=n^{\gamma}_{(t_0,t]}$. Suppose that $k\ge 1$ and let $t_0<t_1<\cdots<t_k\le t$ be all the impulse times within $\gamma \cap (t_0,t]$. Define $z_j=z(t_j)$ for $j=1,\ldots,k$.  Since $z$ is nonincreasing between impulse times and $\bar \psi$ is nondecreasing, it can be proved by induction on $j$ that $z_j\le w_j$ for all $j=0,\ldots,k$, where $\{w_j\}_{j=0}^{\infty}$ is the solution of (\ref{eq:diff}) with $w_0=z_0$. Taking into account that $z(t)\le z(t_k)=z_k$ and that $w_k\le \beta_d(z_0,k)$, it follows that $z(t)\le \beta_d(z_0,k)$. By the forward-in-time uniqueness of the solutions of (\ref{eq:cseq}) we have that $w_{t_0,\zeta}(s)\le w_{t_0,\zeta^*}(s)$ for all $s\ge t_0$ if $0\le \zeta\le \zeta^*$ and that $w_{t_0^*,\zeta^*}(t)=w_{t_0,\zeta}(t)$ for all $t\ge t_0^*$ if $\zeta^*= w_{t_0,\zeta}(t_0^*)$. Then, taking into account the latter and that $w_{t_j,z_j}(t_{j+1})\ge z_{j+1}$ for all $j=0,\ldots, k-1$, it follows that $z_{k}\le w_{t_0,z_0}(t_k)$ and then that $z(t)\le w_{t_0,z_0}(t)\le \beta(z_0,t-t_0)$. When $k=0$, we have that $z(t)\le z_0\le \beta_d(z_0,0)$ and that $z(t)=w_{t_0,z_0}(t)\le \beta(z_0,t-t_0)$. Since $k=n^\gamma_{(t_0,t]}$, we then have that
\begin{align*}
z(t)\le \min\{\beta(z_0,t-t_0),\beta_d(z_0,n^\gamma_{(t_0,t]})\}\quad \forall t\ge t_0.
\end{align*}
By considering $\beta^*=\max\{\beta,\beta_d\}\in \KL$ it follows that for all $t\ge t_0$
\begin{align*}
z(t)&\le \min\{\beta^*(z_0,t-t_0),\beta^*(z_0,n^\gamma_{(t_0,t]})\}
\le \beta^*(z_0,\max\{t-t_0,n^\gamma_{(t_0,t]}\}).
\end{align*}
Since for every $a,b\in\R_{\ge 0}$ it happens that $\max\{a,b\} \ge (a+b)/2$, and given that $\beta^* \in \KL$, then
for all $t\ge t_0$,
\begin{align*}
\beta^*(z_0,\max\{t-t_0,n^\gamma_{(t_0,t]}\}) 
&\le \beta^*\left(z_0,\frac{t-t_0+n^\gamma_{(t_0,t]}}{2}\right)
\end{align*}
Since $\beta^*(\cdot,\cdot/2) \in \KL$, this completes the proof.\qed
\end{proof}


\begin{proof}[Proposition~\ref{prop:main2}] According to Lemma 6.2b) in \citet{mancilla2019uniform}, there exists $\psi^*\in \Ki$ such that $\psi\le \psi^*$ and (\ref{eq:integralr2}) holds with $\psi^*$ in place of $\psi$. Consider the comparison system 
\begin{subequations}
  \label{eq:isc2}
  \begin{align}
    \label{eq:is-ctc2}
    \dot{z}(t) &\in(-\infty,-p(t)\varphi(z(t))], & 
    t\notin \gamma,  \\  
    \label{eq:is-stc2}
    z(t) &\in[0,\psi^*(z(t^-))], & 
    t\in \gamma.
  \end{align}
\end{subequations}
We have that (\ref{eq:isc}) is strongly GUAS uniformly over $\S$ if (\ref{eq:isc2}) is, since every solution of (\ref{eq:isc}) is also a solution of (\ref{eq:isc2}).
Due to Remark~3 in \cite{mancilla2019uniform}, for checking that (\ref{eq:isc2}) is strongly GUAS over $\S$, it suffices to only consider the solutions of the impulsive system
\begin{subequations}
  \label{eq:isc3}
  \begin{align}
    \label{eq:is-ctc3}
    \dot{z}(t) &=-p(t)\varphi(z(t)), & 
    t\notin \gamma,  \\  
    \label{eq:is-stc3}
    z(t) &=\psi^*(z(t^-)), & 
    t\in \gamma,
  \end{align}
\end{subequations}
corresponding to initial times $t_0\ge 0$, initial conditions $z_0\ge 0$ and $\gamma \in \S$. 
From the proof of Theorem 5.2 in \cite{mancilla2019uniform}, it follows that the continuous function $F:(0,\infty)\to \R$
\begin{align}
 F(r)=\int_{1}^{r}\frac{ds}{-\varphi(s)}
\end{align}
is bijective, and that $w_{t_0,w_0}(t)=F^{-1}(F(w_0)+\int_{t_0}^t p(s)ds)$, $t\ge t_0$, is the unique forward-in-time solution of the initial value problem $\dot w=-p(t)\varphi(w)$, $w(t_0)=w_0$, corresponding to $t_0\ge 0$ and $w_0>0$. It also holds that this initial value problem has a unique forward-in-time solution for $t_0\ge 0$ and $w_0=0$, which is the identically zero function. The latter implies that the maximal solutions of (\ref{eq:isc3}) corresponding to nonnegative initial times $t_0$ and initial conditions $w_0$ are unique and forward complete.   
Define
\begin{align*}
  M &:=\inf_{a>0} \int_{\psi^*(a)}^a\frac{ds}{- \varphi(s)}, &N &:=\sup_{t\ge 0}\int_{t}^{t+\theta}p(s)\:ds.
\end{align*}
Note that $M>N\ge 0$.
Set $\delta=M-N>0$ and $\varepsilon_0=\frac{\delta}{2\theta M}$. Since $\S\usubset \Sdn^{u}\left( \frac{1}{\theta} \right)$, there exists $T(\varepsilon_0)>0$ such that $\frac{n^\gamma_{(s,s+t]}}{t}\ge \frac{1}{\theta}-\varepsilon_0$ for all $t\ge T(\varepsilon_0)$, all $s\ge 0$ and all $\gamma\in\S$. Pick $m_0\in \N$ so that $T_0:=m_0\theta \ge T(\varepsilon_0)$ and let $k_0:=\lceil (\frac{1}{\theta}-\varepsilon_0)m_0\theta \rceil$. Note that for any $\gamma\in \S$ the $k_0$-th impulse time after any $t\ge 0$ must belong to the interval $[t,t+T_0]$. Define $\nu(r)=F^{-1}(F(r)+m_0N)$ if $r>0$ and $\nu(0)=0$. We have that $\nu \in \Ki$. 

Let $z$ be a solution of (\ref{eq:isc3}) corresponding to $t_0\ge 0$, $z_0>0$ and $\gamma \in \Gamma$. Since $z$ is nondecreasing between impuse times, and $\psi^*\in \Ki$, it follows that $z(t)>0$ for all $t\ge t_0$.  Let $w_{t_0,z_0}$ as above, then, due to the uniqueness of the solutions of the initial value problems $\dot{w}=-p(t)\varphi(w)$, $w(t_0)=w_0$ and the fact that $\psi^*(r)<r$ for all $r>0$, since $M>0$, we have that $z(t)\le w_{t_0,z_0}(t)=F^{-1}(F(w_0)+\int_{t_0}^t p(s)ds)$ for all $t\ge t_0$. In particular, taking into account that $\int_{t_0}^{t_0+T_0} p(s)ds\le m_0N$, it follows that
\begin{align}\label{eq:nu}
 z(t) \le \nu(z_0)\quad \forall t\in [t_0,t_0+T_0]. 
\end{align}
Let $t_1<t_2<\cdots$ be the sequence of impulse times after time $t_0$, that is $\{t_k\}_{k=1}^{\infty}=\gamma \cap (t_0,\infty)$. Note that necessarily $t_{k_0}\in [t_0,t_0+T_0]$. From (\ref{eq:isc3}) it follows that
\begin{align} 
  \int\limits_{z(t_0)}^{z(t_{k_0})} \frac{ds}{-\varphi(s)} &= \sum_{k=0}^{k_0-1} \int_{z(t_k)}^{z(t_{k+1}^-)} \frac{ds}{-\varphi(s)}+ \sum_{k=1}^{k_0}\int_{z(t_k^-)}^{z(t_{k})}\frac{ds}{-\varphi(s)} 
=\sum_{k=0}^{k_0-1}\int_{z(t_k)}^{z(t_{k+1}^-)}\frac{ds}{-\varphi(s)}+ \sum_{k=1}^{k_0}\int_{z(t_k)^-}^{\psi^*(z(t_{k}))}\frac{ds}{-\varphi(s)} \notag \\
&\le \int_{t_0}^{t_{k_0}} p(s)ds - k_0M \le m_0N-\left (\frac{1}{\theta}-\varepsilon_0\right)m_0\theta M 
  \le -\frac{1}{2}m_0\delta. \label{eq:deltad2}
 \end{align}
In consequence $z(t_{k_0})<z(t_0)$. Taking into account (\ref{eq:nu}), (\ref{eq:deltad2}), and applying the preceding reasoning with $t_{\ell k_0}$ as $t_0$ and $t_{(\ell+1) k_0}$ as $t_{k_0}$, we have that $\{z(t_{\ell k_0})\}_{\ell=0}^{\infty}$ is a decreasing sequence, and that for all $t\in [t_{\ell k_0},t_{(\ell+1) k_0}]$, we have that $z(t)\le \nu(z(t_{\ell k_0}))\le \nu(z_0)$ since $[t_{\ell k_0},t_{(\ell+1) k_0}]\subset [t_{\ell k_0},t_{\ell k_0}+T_0]$. 
Taking into account that (\ref{eq:deltad2}) holds with $t_{\ell k_0}$ and $t_{(\ell+1) k_0}$ in place of $t_0$ and $t_{k_0}$, it follows that
\begin{align}
 \int_{z(t_0)}^{z(t_{\ell k_0})}\frac{ds}{-\varphi(s)}\le -\frac{1}{2}\ell m_0\delta,
\end{align}
and therefore $z(t_{\ell k_0})\le F^{-1}(F(z(t_0))-\frac{1}{2}\ell m_0\delta)$. Define
\begin{align*}
  \beta_1(r,s) &:=
                 \begin{cases}
                   F^{-1}(F(r)-\frac{1}{2}m_0\delta s) &\text{if $r>0$ and $s\ge 0$},\\
                   0 &\text{if $r=0$ and $s\ge 0$,}
                 \end{cases}
\end{align*}
and let $\beta_2=\nu\comp\beta_1$. Then $\beta_2\in \KL$ and $z(t)\le \beta_2(z(t_0),\ell)$ if $t\in [t_{\ell k_0},t_{(\ell+1) k_0})$. Note that for such a value of $t$, $\ell k_0\le n^\gamma_{(t_0,t]} < (\ell+1)k_0$ and then  $\ell=\lfloor n^\gamma_{(t_0,t]}/k_0 \rfloor$. In consequence,
\begin{align}
 z(t)\le \beta_2(z(t_0),\lfloor n^\gamma_{(t_0,t]}/k_0 \rfloor),\quad \forall t\ge t_0.
\end{align}
Define $\beta_3:\R_{\ge 0}\times \R_{\ge 0}\to \R_{\ge 0}$ via
\begin{align*}
  \beta_3(r,s) = 
  \begin{cases} 
    (2-\frac{s}{k_0})\beta_2(r,0) & 0\le s<k_0,\;r\ge 0, \\
    \beta_2(r,\frac{s}{k_0}-1),   & s\ge k_0,\;r\ge 0.
  \end{cases}
\end{align*}
Then $\beta_3\in \KL$ and 
\begin{align}
  \label{eq:beta3}
  z(t)\le \beta_3(z(t_0),n^\gamma_{(t_0,t]}),\quad \forall t\ge t_0.
\end{align}
Since $n^\gamma_{(t_0,t]} \ge \left (\frac{1}{\theta}-\varepsilon_0\right ) (t-t_0)$ if $t-t_0\ge T_0$, it follows that
\begin{align} 
  \label{eq:bound4}
  \frac{1}{2}n^\gamma_{(t_0,t]} \ge \max\{\kappa (t-t_0)-\kappa T_0,0\}\quad\forall t\ge t_0,
\end{align}
where $\kappa=\frac{1}{2}\left (\frac{1}{\theta}-\varepsilon_0\right )$.
The combination of (\ref{eq:beta3}) with (\ref{eq:bound4}) yields that, for all $t\ge t_0$,
\begin{align*} \label{eq:bound5}
 z(t)\le \beta_3\left (z(t_0),\frac{n^\gamma_{(t_0,t]}}{2}+\max\{\kappa (t-t_0)-\kappa T_0,0\}\right ).
\end{align*}
From the latter and proceeding as in the last part of the proof of Theorem 5.2 in \cite{mancilla2019uniform} it follows that there exits $\beta\in \KL$ so that 
\begin{align*}
 z(t)\le \beta\left (z(t_0),t-t_0 + n^\gamma_{(t_0,t]} \right )\quad \forall t\ge t_0.
\end{align*}
This establishes the required strong GUAS property.\qed
\end{proof}

\section{Examples}
\label{sec:examples}

\subsection{Destabilizing impulses}
\label{sec:destab-imp}

Consider a scalar and single-input impulsive system of the form (\ref{eq:is}) with time-invariant flow and jump maps, defined as
\begin{align*}
  f(t,\xi,\mu) = \bar f(\xi,\mu) &=
  \begin{cases}
    -\xi + \sqrt{2}\mu &\text{if }|\xi| \le \sqrt{2},\\
    -\frac{\xi^3}{2} + |\xi| \mu &\text{if }|\xi| > \sqrt{2},
  \end{cases} &
  g(t,\xi,\mu) = \bar g(\xi,\mu) &=
                 \begin{cases}
                   2\sqrt{2}\xi^2 \mu &\text{if }|\xi| \le \sqrt{2},\\
                   \xi^3 &\text{if }|\xi| > \sqrt{2},
                 \end{cases}
\end{align*}
Note that under zero input, the flow is stabilizing and the jumps are destabilizing. 
We would like to determine a class of impulse-time sequences over which this system is ISS in the standard sense. We thus consider $h^o(t,\xi) = h(t,\xi) = |\xi|$ and search for an ISS Lyapunov function candidate. Taking $V(t,\xi) = \bar V(\xi) = \xi^2/2$, we compute
\begin{align*}
  D_f^+ \bar V(\xi,\mu) &=
  \begin{cases}
    -\xi^2 + \sqrt{2}\xi\mu &\text{if }|\xi| \le \sqrt{2},\\
    -\frac{\xi^4}{2} + \xi^2 \mu &\text{if }|\xi| > \sqrt{2},
  \end{cases} 
  &\bar V(\xi + \bar g(\xi,\mu)) &=
                          \begin{cases}
                            (\xi + 2\sqrt{2}\xi^2 \mu)^2/2 &\text{if }|\xi| \le \sqrt{2},\\
                            (\xi + \xi^3)^2/2 &\text{if }|\xi| > \sqrt{2}.
                          \end{cases}
\end{align*}
If it happens that
\begin{align*}
  |\mu| \le \eta(|\xi|) :=
  \begin{cases}
    \dfrac{|\xi|}{2\sqrt{2}} & \text{if }|\xi| \le \sqrt{2},\\[3mm]
    \dfrac{|\xi|^2}{4} & \text{if }|\xi| > \sqrt{2},
  \end{cases}
\end{align*}
then
\begin{align*}
  D_f^+ \bar V(\xi,\mu) &\le
  \begin{cases}
    -\xi^2/2 = -\bar V(\xi)  & \text{if }|\xi| \le \sqrt{2},\\
    -\xi^4/4 = -\bar V(\xi)^2 & \text{if }|\xi| > \sqrt{2},
  \end{cases}\\
  \bar V(\xi + \bar g(\xi,\mu)) &\le (|\xi| + |\xi|^3)^2/2 = \psi(\bar V(\xi)), &
  \text{with }\psi(s) &= s + 4s^2 + 4s^3.
\end{align*}
Since $\eta \in \Ki$, we may define $\chi = (\eta^{-1})^2/2 \in \Ki$. Define also $\varphi,p : \R_{\ge 0} \to \R_{\ge 0}$ via
 \begin{align*}
  \varphi(s) &:=
  \begin{cases}
    s &\text{if }s \le 1,\\
    s^2 &\text{if }s>1,
  \end{cases} & p(s) &\equiv 1.
\end{align*}
Also, we have $\bar V(\xi) \le \chi(|\mu|)$ if and only if $|\xi| \le \eta^{-1}(|\mu|)$. As a consequence, whenever $\bar V(\xi) \le \chi(|\mu|)$ then
\begin{align*}
  \bar V(\xi+g(t,\xi,\mu)) \le \left(\eta^{-1}(|\mu|) + [\eta^{-1}(|\mu|)]^3\right)^2/2 =: \pi(|\mu|).
\end{align*}
It follows that $V$ is an ISS Lyapunov function candidate with $\chi,\pi,\varphi$ and $\psi$ as defined. To apply Theorem~\ref{thm:main}, we compute
\begin{align*}
  \int_a^{\psi(a)} \frac{ds}{\varphi(s)} =
  \begin{cases}
    \ln\left(\frac{\psi(a)}{a}\right) &\text{if }\psi(a) \le 1,\\
    \frac{1}{a}-\frac{1}{\psi(a)} &\text{if }a\ge 1,\\
    \ln\left(\frac{1}{a}\right) + 1 - \frac{1}{\psi(a)} &\text{if }a < 1 < \psi(a).
  \end{cases}
\end{align*}
and we obtain
\begin{align*}
  M &= \sup_{a>0} \int_a^{\psi(a)} \frac{ds}{\varphi(s)} = \int_{0.5}^{\psi(0.5)} \frac{ds}{\varphi(s)} \approx 1.1931,\\
  N &= \inf_{t\ge 0} \int_t^{t+\theta} p(s) ds = \theta.
\end{align*}
Application of Theorem~\ref{thm:main} gives that the system is strongly ISS over any $\S \subset \Gamma$ such that $\S \usubset \Sup(1/\theta)$ whenever $\theta > M$. 

\subsection{Stabilizing impulses}
\label{sec:ex-stab-imp}

We consider a scalar time-varying impulsive system with no inputs and flow and jump maps given by (\ref{eq:is}) with
\begin{align*}
  f(t,\xi,\mu) &= \tanh(t) \tanh(\xi),\\
  g(t,\xi,\mu) &=
                 \begin{cases}
                   -\xi + \xi^3/2 &\text{if }|\xi|\le 1,\\
                   -\xi + \xi^{1/3}/2 &\text{if }|\xi| > 1.
                 \end{cases}
\end{align*}
The flow map is a time-varying multiple of that in Example~1 of \citet{Automatica2019} and the jump map coincides. We would like to know the largest class of impulse-time sequences over which the impulsive system is GUAS, according to Theorem~\ref{thm:main2}. Since the system has no inputs, we take $m=0$ and consider $u\in\R^0$. Since standard GUAS is considered, we take $h^o(t,\xi) = h(t,\xi) = |\xi|$. Taking $V(t,\xi) = \bar V(\xi) = |\xi|$, it follows that $V$ is a $(h^o,h)$-ISS Lyapunov function candidate with 
\begin{align*}
  p(t) &= \tanh(t), & \varphi(s) &= -\tanh(s),
  &\psi(s) &=
            \begin{cases}
              s^3/2 &\text{if }0\le s\le 1,\\
              s^{1/3}/2 &\text{if }s>1,
            \end{cases}
\end{align*}
and arbitrary $\chi,\pi\in\Ki$. We have
\begin{align*}
  \sup_{t\ge 0} \int_t^{t+\theta} \tanh(s) ds = \theta,
\end{align*}
and as in Example~1 of \citet{Automatica2019},
\begin{align*}
  \inf_{a>0} \int_{\psi(a)}^a \frac{ds}{-\varphi(s)} = \ln(1+e)-\frac{1}{2} \approx 0.81.
\end{align*}
Also, (\ref{eq:nofet}) is satisfied. Theorem~\ref{thm:main2} ensures that the impulsive system is strongly $(h^0,h)$-ISS, hence in this case strongly GUAS, over any class $\S\subset\Gamma$ satisfying $\S\usubset\Sdn(1/\theta)$ provided that $0 < \theta < \ln(1+e)-\frac{1}{2}$. Therefore, strong GUAS is ensured over any uniform subset of the class of impulse-time sequences with i.f. eventually uniformly lower bounded by $1/\theta$. For comparison with Example~1 of \citet{Automatica2019}, take $\theta = 0.8$. According to Lemma~\ref{lem:relats}, items~\ref{item:fixed-euup}) and~\ref{item:pers-pert}), the classes $\S$ satisfying $\S\usubset\Sdn(1/\theta)$ are much larger than what is covered by the results in \citet{Automatica2019}. In particular, the impulsive system is GUAS over sequences with consecutive impulses much more separated than the value $3/2$ in Example~1 of \citet{Automatica2019} provided that the i.f. eventually becomes uniformly lower bounded by $1/\theta = 1.25$. In addition, note that the results of \citet{Automatica2019}, by contrast to the current ones, only ensure standard, i.e. weak, GUAS that is not uniform w.r.t. initial time.

\section{Conclusions}
\label{sec:conclusions}

We have developed novel sufficient conditions for the weak or strong ISS of nonlinear time-varying impulsive systems with inputs, employing a two-measure framework. These conditions generalize, extend, and strengthen many existing results by ensuring ISS that holds uniformly over impulse-time sequences having eventually uniformly bounded impulse frequency. We show that the considered classes of impulse-time sequences are broader than most other sequence classes considered in the literature. In particular, sequences with fixed and average dwell times, as well as sequences where the impulse frequency achieves uniform convergence to a limit (superior or inferior) are all covered. 

\section*{Acknowledgments}

Work partially supported by Agencia Nacional de Promoci\'on Cient\'{\i}fica y Tecnol\'ogica, Argentina, under grant PICT 2018-01385.

\section*{References}

\bibliography{references}
\end{document}